\documentclass[conference]{IEEEtran}
\IEEEoverridecommandlockouts
\IEEEoverridecommandlockouts
\usepackage{cite}
  \usepackage[dvips]{graphicx}
  \DeclareGraphicsExtensions{.eps}
\usepackage[cmex10]{amsmath}

\usepackage{amssymb}
\usepackage{latexsym,amsfonts}
\usepackage{algorithmic}

\usepackage{array}

\usepackage[caption=false,font=footnotesize]{subfig}
\usepackage{fixltx2e}

\usepackage[dvips]{color}

\newcounter{MYtempeqncnt}

\newtheorem{theorem}{Theorem}
\newtheorem{lemma}[theorem]{Lemma}

\newtheorem{corollary}[theorem]{Corollary}

\newcommand{\E}{\mathbb{E}}

\begin{document}
\title{Chernoff Bounds for Analysis of Rate-Compatible Sphere-Packing with Numerous Transmissions
\thanks{This research was carried out in part at the Jet Propulsion Laboratory, California Institute of Technology, under a contract with NASA.
}
}


\author{\IEEEauthorblockN{Tsung-Yi Chen}
\IEEEauthorblockA{Department of Electrical Engineering\\
University of California, Los Angeles\\
Los Angeles, California 90024
}
\and
\IEEEauthorblockN{Dariush Divsalar}
\IEEEauthorblockA{Jet Propulsion Laboratory\\
California Institute of Technology\\
Pasadena, California 91109
\and
\IEEEauthorblockN{Richard D. Wesel}
\IEEEauthorblockA{Department of Electrical Engineering\\
University of California, Los Angeles\\
Los Angeles, California 90024
}
}
}


\maketitle

\begin{abstract}
\boldmath
Recent results by Chen et al. and Polyanskiy et al. explore using feedback to approach capacity with short blocklengths. This paper explores Chernoff bounding techniques to extend the rate-compatible sphere-packing (RCSP) analysis proposed by Chen et al. to scenarios involving numerous retransmissions and different step sizes in each incremental retransmission.  Williamson et al. employ exact RCSP computations for up to six transmissions.  However, exact RCSP computation with more than six retransmissions becomes  unwieldy because of joint error probabilities involving numerous chi-squared distributions. This paper explores Chernoff approaches for upper and lower bounds to provide support for computations involving more than six transmissions.

We present two versions of upper and lower bounds for the two-transmission case. One of the versions is extended to the general case of $m$ transmissions where $m \geq 1$. Computing the general bounds requires minimization of exponential functions with the auxiliary parameters, but is less complex and more stable than multiple rounds of numerical integration. These bounds also provide a good estimate of the expected throughput and expected latency, which are useful for optimization purposes. 
\end{abstract}

\section{Introduction}

\subsection{Previous Work}
It is well known that feedback can significantly improve the error exponent, \cite{Chang_1956, Kramer_1969, Zig_1970}. Under AWGN channel with noiseless feedback, the Shalkwijk and Kailath  (SK) coding scheme achieves error probability with doubly exponential decay\cite{Schalkwijk_1966_1}\cite{Schalkwijk_1966_2}. The SK scheme can be derived elegantly with the Elias result as shown in \cite{Gallager_2010}. The SK scheme is simple and efficient but requires full knowledge of the signal seen by the receiver to be communicated to the transmitter via feedback. 

On the other hand, Polyanskiy et al. \cite{PolyIT11} show that full information through feedback is not necessary to achieve throughput close to capacity with low latency. Chen et al. \cite{Chen_2010_ITA} also show by simulation that a simple incremental redundancy scheme with feedback will allow a convolutional code with blocklength less than $200$ to perform close to an LDPC code with blocklength close to $2000$. 

The Rate-Compatible Sphere-Packing (RCSP) analysis was first proposed in \cite{Chen_2011_ICC} as an analytic tool to characterize the capacity-achieving potential of Hybrid ARQ systems.  Both analysis and simulation results show that a simple feedback scheme using (ACK/NACK) with incremental redundancy allows the system to achieve $90$\% of the capacity with an average latency less than $100$ symbols. 

Achievability and converse bounds of variable length coding shown in \cite{PolyIT11} reveal similar results of significant latency reduction when a noiseless feedback is presence. In particular, an example for the binary symmetric channel (BSC) shows that achieving $90$\% of the capacity requires only less than $200$ symbols. 
 
\subsection{Main Contribution}

Chen et al. \cite{Chen_2011_ICC} provide an approximation formula to compute the joint error probability of multiple transmissions with fixed step sizes.  Computing the probability by numerical integration is also possible for small number of transmissions \cite{Williamson_ISIT_2012}. As the number of transmissions grows, however, both the approximation and numerical method become unwieldy for optimization purposes. 

This paper provides lower and upper bounds on the relevant joint error probability for RCSP. The upper and lower bounds are given as infimums of closed form functions that require much less computation power. For the two-transmission case, two versions of upper and lower bounds are derived. The version where Chernoff bounds are used can be generalized to the $m$-transmission case. 

The results can be translated into upper and lower bounds on the expected throughput and expected latency. These bounds provide tight estimates of the expected latency and can be used to optimize the transmission rate and blocklengths for practical incremental redundancy schemes for general $m$ transmissions. Examples in Section \ref{sec:Examples} show that relaxing the bounds to several pairs of joint error events with a suboptimal but closed form auxiliary parameter already gives very sharp results.

\section{Problem Statement}
\label{sec:PS}
In this paper we consider a communication system under AWGN channel with a noiseless feedback. The use of the feedback link in our system is minimal: sending one bit of information for each block of forward transmission to confirm whether the message is received correctly. If the transmission is not successful, the transmitter will retransmit a block of non-repetitive incremental redundancy. The transmitter will attempt up to $m$ transmissions (including the initial transmission). After the $m$th transmission the transmitter will restart from the initial transmission. 

The receiver uses a spherical bounded distance decoder defined as follows: consider an $(M, n)$ code where $M$ is the number of messages and $n$ is the blocklength. The decoder maps the received sequence $Y_1^n$ to the codeword that is within the decoding radius $r$ (in terms of Euclidean distance). If there are more than one codewords or there is no codeword within the distance, the decoder will declare an error. 

For systems with feedback where only limited number of transmissions can be permitted, RCSP analysis (which assumes that the decoding radius $r$ is what would be achieved by ideal sphere packing) provides practical guidance of the optimal transmission rate and blocklength.  This was demonstrated in \cite{Williamson_ISIT_2012}.  One of the issues in using RCSP to optimize transmission rates and blocklengths is the complexity of performing numerical integration to compute the joint error probability. In \cite{Williamson_ISIT_2012}, exact computations were made, but this was only possible for up to six transmissions.  In \cite{Chen_2011_ICC},  an approximation formula based on the i.i.d. assumption gives an accurate estimate when the step sizes are large (which supports the i.i.d. assumption). 

This paper gives tight upper and lower bounds that allow a large number of transmissions and a relatively small step size.  These bounds are in closed form (or are optimizations of closed form functions). The main results in the following sections are expressed in terms of the decoding radius $r_i$ of the codeword received at the $i$th retransmission. These results can then be evaluated by replacing $r_i$ with proper expression according to different assumptions. We always keep the cumulative distribution functions (CDFs) although they can also be bounded if desired.

\section{Main Results}
\label{sec:Results}
\subsection{Bounds on Joint Error Event: Two-Transmission}
This section summarizes the upper and lower bounds of the joint error probability using a spherical bounded-distance decoder under AWGN channel. Some of the proofs are given in the Appendix.

Let $N_1 = I_1$ be the blocklength of the initial transmission and $I_i$ be the blocklength of the incremental redundancy transmitted at the $i$th transmission. The number of accumulated symbols at $i$th transmission is then $N_i = N_{i-1} + I_{i}$. Denote the bounded-distance decoding radius for the $i$th transmission as $r_i$. 


Suppose without loss of generality that the noise samples $z_i$ are i.i.d. and $z_i\sim \mathcal{N}(0, 1)$. The error event $\zeta_i$ of the $i$th transmission is given as $\zeta_i = \left\{\sum\limits_{i = 1}^{N_i}z_i^2 > r_i^2 \right\}$. The probability of each error event is simply the tail of a chi-square random variable: $\Pr(\zeta_i) = \Pr\left(\chi_{N_i}^2 > r_i^2\right)$. 

Because of the dependency between $\zeta_i$'s, the probabilities of the joint events can only be expressed by integration. Take the two-transmission case for example, the joint error probability is given as 

\begin{align}
	&\Pr(\zeta_1\cap\zeta_2) 
	\label{eqn:Integrate1}
	= \int_{r_1^2}^{\infty} \Pr\left(\chi_{I_2}^2 > r_2^2-t\right)f_{\chi_{I_1}^2}(t) dt
	\\ 
	\label{eqn:Integrate2}
	=& \int_{r_1^2}^{r_2^2} \Pr\left(\chi_{I_2}^2 > r_2^2-t\right)f_{\chi_{I_1}^2}(t) dt +  \Pr\left(\chi_{I_1}^2 > r_2^2\right)
\end{align}
where $ f_{\chi_{n}^2}(t)$ is the density function of a chi-square distribution with $n$ degrees of freedom.

We first summarize the two versions of upper and lower bounds for the two-transmission case. The first version of the upper and lower bounds uses the Chernoff bound.  The following lemma states these upper and lower bounds for the two-transmission case.
\begin{lemma}
\label{lem:ChernoffTwoTrans}
\begin{align}
&\Pr(\zeta_1 \cap \zeta_2) \leq \inf\limits_{0\leq u < 1/2} \frac{e^{-u r_2^2}\Pr\left(\chi_{I_1}^2 > (1-2u)r_1^2\right)}{(1-2u)^{N_2/2}}
\\
&\Pr(\zeta_1 \cap \zeta_2) \geq \max\left(\Pr(\zeta_1)-w_1, \Pr(\zeta_2)-w_2 \right)
\\
&\text{where } w_1 = \inf\limits_{v\geq 0}\frac{e^{v r_2^2}\Pr\left(\chi_{I_1}^2 > (1+2v)r_1^2\right)}{(1+2v)^{N_2/2}},
\\ 
&w_2 = \inf\limits_{0\leq v\leq 1/2}\frac{e^{-v r_2^2}\Pr\left(\chi_{I_1}^2 \leq (1-2v)r_1^2\right)}{(1-2v)^{N_2/2}}.
\end{align}
\end{lemma}

Further bounding the expressions of  Lemma \ref{lem:ChernoffTwoTrans} is possible and can yield  convex functions, but the $u$ that optimizes these convex functions does not necessarily give the best bound in Lemma \ref{lem:ChernoffTwoTrans}.  Instead, we use a suboptimal but insightful parameter $u^* = (1 - N_2/r_2^2)/2$. Let $c_2 = r_2^2/N_2$. The upper bound, for example, then becomes 
\begin{align}
\label{eqn:sumoptimal_u}
\exp\left(-N_2(c_2 - 1 - \ln c_2)/2\right)\Pr\left(\chi_{I_1}^2 > r_1^2/c_2\right).
\end{align}
Assuming perfect sphere-packing (see Section \ref{sec:SpherePacking}), the radius-adjusting parameter $c_2$ is always greater than $1$ ( hence $u^* < 1/2$ ) if the code rate is less than capacity. Note that our choice of $u^*$ does optimize the Chernoff upper bound for $\Pr(\zeta_2)$ and gives the expression $\exp\left(-N_2(c_2 - 1 - \ln c_2)/2\right)$. Equation \eqref{eqn:sumoptimal_u} says that  $\Pr(\zeta_1, \zeta_2)$ is approximately $\Pr(\zeta_2)$  multiplied by the probability of the first error event but with squared radius $r_1^2$ divided by the factor $c_2$.

It is observed in \cite{Chen_2011_ICC} that the first few transmissions should have rates slightly above capacity to achieve the best expected throughput with feedback. The above Chernoff bounds give trivial results when the rate is above capacity, hence we provide a second version of the upper and lower bounds based on the results by Inglot\cite{Inglot2006}.  We first state the theorem given by Inglot:

\begin{theorem}[Inlgot\cite{Inglot2006}]
\label{thm:Inglot}
Let $\chi^2_k$ denote a random variable with central chi-square distribution and $k$ degrees of freedom. For $k \geq 2, r > k-2$,
	\begin{align}
	\label{eqn:TailofChiSqr}
	\frac{1}{2}\mathcal{E}_k(r)\leq \Pr(\chi^2_k > r) \leq \frac{1}{\sqrt{\pi}}\frac{r}{r-k+2}\mathcal{E}_k(r)
	\end{align}
	where $\mathcal{E}_k(r) = \exp\left\{\frac{-1}{2}\left[ r - k - (k-2)\log(r/k) + \log k \right] \right\}$
\end{theorem}

Theorem \ref{thm:Inglot} gives the following result for the two-transmission case. To simplify the equation let $p \equiv \Pr\left(\chi_{I_1}^2 > r_2^2\right)$.

\begin{theorem}
\label{thm:InglotTwoTrans}
\begin{align}
&p + \frac{\sqrt{\pi}K}{2} \int_{r_1^2}^{r_2^2} (r_2^2 - t)^{I_2/2-1}t^{I_1/2-1}dt \leq \Pr(\zeta_1 \cap \zeta_2)\nonumber
\\
 &\leq  p + \inf\limits_{\delta \in (\underline{\delta} ,\overline{\delta})} K\int_{r_1^2}^{(1-\delta)r_2^2} g(t) dt  + 
 \int_{(1-\delta)r_2^2}^{r_2^2} f_{\chi_{I_1}^2}(t) dt 
\end{align}
where  $\overline{\delta}, \underline{\delta},g(t)$ and $K$ are described in detail in the Appendix.
\end{theorem}

Although Theorem \ref{thm:InglotTwoTrans} cannot be generalized to the $m$-transmission case, numerical results show that the joint error probability on two events already gives surprisingly tight bounds (details are discussed in Section \ref{sec:Examples}). Hence Theorem \ref{thm:InglotTwoTrans} may still be useful to obtain even tighter bounds especially when the rate is slightly higher than the capacity. The integral in Theorem \ref{thm:InglotTwoTrans} can be expressed in terms of an Appell hypergeometric function of two variables in closed form. See details in \cite{ChenDraft2012}.
\subsection{Bounds on the Joint Error Event: $m$ Transmissions}
\label{sec:GeneralCase}
In the general case where $m$ transmissions are allowed, there are $m-1$ step sizes $I_2, \dots, I_m$ and $N_1 = I_1, N_i = N_{i-1}+ I_i $ for $i > 1$. The joint error probability can be expressed by \eqref{eqn:Pzeta_i}.

The following results give the upper and lower bounds based on Chernoff bounds for the $m$-transmission joint error probability:

\begin{theorem}
\label{thm:GeneralUpper}
Let $u_i<1/2, i = 1, 2, \dots, m$ be the parameters for each use of Chernoff bound in the integral. Define $h_i, g_i(u_1^m)$ by the following recursion:
\begin{align*}
h_1 &= u_1, h_i = h_{i-1}+u_i(1-2h_{i-1}),
\\
g_1 &= e^{-u_1r_m^2}(1-2u_1)^{\frac{-I_{m}}{2}},
\\
g_i &= g_{i-1}e^{-u_i(1-2h_{i-1})r_{m-i+1}^2}\left(1-2h_{i-1}\right)^{\frac{-I_{m-i+1}}{2}}.
\end{align*}
Note the property that $1-2h_{i} = \prod_{j \leq i} (1-2u_j)$. We have
\small
\begin{align}
&\Pr\left(\bigcap\limits_{i = 1}^m \zeta_i \right) 
\leq \inf\limits_{u_1^m} \frac{g_{m-1}(u_1^m)\Pr\left(\chi_{I_1}^2 > (1-2h_{m-1})r_1^2\right) }{(1-2h_{m-1})^{I_1/2}}
\end{align}
\normalsize
\end{theorem}

Several versions of lower bounds can be obtained by different expansions and the recursion formulas follow closely to those in Theorem \ref{thm:GeneralUpper}. The following corollary gives an example of one specific expansion that yields a lower bound in a recursive fashion. The other formulas are omitted due to space limitations. See another example in Section~\ref{sec:Examples}.
\begin{corollary}
\label{cor:GeneralLower}
Write $\Pr\left(\cap_{1\leq i \leq m } \zeta_i \right) = \Pr\left(\cap_{2\leq i \leq m } \zeta_i \right) - \Pr\left(\cap_{2\leq i \leq m } \zeta_i \cap \zeta_1^c\right)$.  With the same recursion as in Theorem \ref{thm:GeneralUpper},  the lower bound is given in \eqref{eqn:GeneralLowerBound}. 

%
\end{corollary}

\begin{figure*}
\setcounter{MYtempeqncnt}{\value{equation}}
\setcounter{equation}{13}
\small
\begin{align}
	\label{eqn:Pzeta_i}
	\Pr\left(\bigcap \limits _{j=1}^{m} \zeta_j\right)
	= \int_{r_1^2}^{\infty} \int_{r_2^2-t_1}^{\infty} \dots 
	\int_{r_{m-1}^2- \sum\limits_{j=1}^{m-2}t_j}^{\infty}  
	f_{\chi_{N_1}^2}(t_1) \dots f_{\chi_{I_{i-1}}^2}(t_{m-1}) 
	\Pr\left(\chi_{I_{m}}^2 > r_m^2 - \sum_{j=1}^{m-1}t_j\right) dt_{m-1} \dots dt_1
\end{align}
	\hrulefill
	\vspace*{-9pt}
	
\begin{align}
	\label{eqn:GeneralLowerBound}
	\Pr\left(\bigcap\limits_{i = 1}^m \zeta_i \right) 
	\geq 	\max\left(0, \Pr\left(\bigcap\limits_{i = 2}^m \zeta_i \right) 
	- \inf\limits_{u_1^m} \frac{g_{m-1}(u_1^m)\Pr\left(\chi_{I_1}^2 \leq (1-2h_{m-1})r_1^2\right) }{(1-2h_{m-1})^{I_1/2}}\right)
\end{align}
	\hrulefill
	\vspace*{-9pt}
\end{figure*}
\normalsize
%
\section{Application to RCSP}
\label{sec:SpherePacking}

For the RCSP analysis, we usually assume that at each decoding attempt, each received codeword will pack the sphere generated by the power constraint according to perfect sphere-packing.  We also consider a more pessimistic assumption using Minkowski's lower bound. 

Consider an $(M, n)$ code on the AWGN channel and let the SNR be $\eta$. Assume without loss of generosity that each noise sample has a unit variance. Then the average power of a received codeword is less than $P = n(1 + \eta)$.  Sphere-packing seeks a codebook that has $M$ codewords that represent the center of spheres (possibly overlapped) that are packed inside the $n$-dimensional ball with radius $\sqrt{n(1 + \eta)}$.


\subsection{Pessimistic Sphere-Packing Under AWGN Channel}

This subsection uses the classic result by Minkowski in the sphere-packing argument: the packing density $\phi \geq c 2^{-n}$ for some constant $c>1$ in $\mathbb{R}^n$. We use Minkowski's result to simplify our analysis even though the best known result scales as $nc_n 2^{-n}$ asymptotically \cite{Sloane2002}. 
	
The following theorem states that when $\eta > 1$ the decoding time is finite a.s. and that the expected latency is also finite.
\begin{theorem}
	\label{thm:DecodingTime}
	Assume that there exists a rate-compatible code with radii $r_i$ that at least achieves the packing density $c2^{-n}$ in $\mathbb{R}^n$ and the $\eta > 1$. Let $N_i$, a subsequence of $\mathbb{N}$, be the blocklengths at each decoding attempt. Let the decoding time (also stopping time w.r.t. the natural filtration generated by $\{Z_i\}_i$) $\tau = \inf\limits_i \left\{\chi_{N_i}^2 < r_i^2 \right\}$ and let $L$ be the latency, then $\Pr(\tau\text{ is finite}) = 1$ and $\E L < \infty$.
\end{theorem}
	
	Since Minkowski's result is a lower bound on the packing density, Theorem \ref{thm:DecodingTime} also holds under the perfect packing assumption in the next subsection.

\subsection{Optimistic Sphere-Packing Under AWGN Channel} 
\label{sec:Examples}
This subsection briefly reviews the argument of obtaining optimistic sphere-packing radii used in \cite{Chen_2011_ICC} and provides numerical examples based on the sphere-packing radii. The largest sphere-packing radius perfectly packs $M$ spheres into the outer sphere. With this sphere-packing in mind, a conservation of volume argument yields the following inequality:
\begin{align*}
 \text{Vol}({\text{Inner sphere}}) \le \frac{{ \text{Vol}({\text{Outer sphere}})}}{{M}} 
\Rightarrow r_i^2 \le \frac{N_i(1 + \eta )}{M^{2/N_i}} .
\end{align*}

Based on the optimistic sphere-packing assumption, we give some examples of applying the bounds on joint error probability to obtain the latency versus throughput curve. 

Using the zero-error coding scheme described at the beginning of Section \ref{sec:PS}, the expected latency $\E L$ and throughput $\E R_t$ are given as follows (assuming $\Pr(\zeta_0) = 1$)
\begin{align}
\setcounter{equation}{10}
\E L = \frac{\sum\limits_{i = 1}^m I_i\Pr\left(\cap_{j \leq i-1} \zeta_j \right)}{1-\Pr\left(\cap_{j \leq m} \zeta_j \right)}, \E R_t = \log M/ \E L.
\end{align}

We apply Theorem \ref{thm:GeneralUpper} and its corollary to derive a tight lower bound on the joint error event. Rewrite the joint error event as $\bigcap_{i\leq m}\zeta_i =  \zeta_m\setminus \zeta_m\cap\left(\bigcap_{i\leq m-1}\zeta_i\right)^c$, which comes from the disjoint union of $\zeta_m$: $\cap_{i\leq m}\zeta_i \cup \zeta_m\left(\cap_{i\leq m-1}\zeta_i\right)^c =  \zeta_m$. By De Morgan's law
\begin{align}
 &\Pr\left(\bigcap_{i=1}^{m}\zeta_i\right) = \Pr(\zeta_m) - \Pr\left(\zeta_m\cap \bigcup_{i= 1}^{m-1}\zeta_i^c\right)
 \\
 \label{eqn:LowerEg}
 &\geq  \Pr(\zeta_m) - \sum\limits_{i = 1}^{m-1}\Pr\left(\zeta_m \cap \zeta_i^c\right)
\end{align}
where the last inequality can be seen as the union bound on the second term of the first equality.
Setting the parameter $u = 1/2 - N_m/(2r_m^2+2\log_2 M)$ in Theorem \ref{thm:GeneralUpper} with $m=2$ gives a fairly tight lower bound despite being suboptimal, as shown in Fig \ref{fig:m5compare}. We comment here that applying Theorem \ref{thm:InglotTwoTrans} may yield an even better bound but the evaluation of the bound is slightly more complex. 

To obtain an upper bound we rewrite the joint error probability up to $j$th transmission ($j \leq  m$) as 
\begin{align}
\setcounter{equation}{15}
 &\Pr\left(\bigcap_{i=1}^{j}\zeta_i\right) 
 = \Pr\left(\bigcap_{i=1}^{j}\zeta_i \cap \zeta_m\right) +\Pr\left(\bigcap_{i=1}^{j}\zeta_i \cap \zeta_m^c\right)
 \\\label{eqn:UpperEg}
 &\leq  \Pr(\zeta_m, \zeta_j) + \Pr(\zeta_j\cap \zeta_{j-1}) - \Pr(\zeta_j\cap\zeta_{j-1}\cap\zeta_m)
\end{align}
Applying Theorem \ref{thm:GeneralUpper} to the first two terms and Corollary \ref{cor:GeneralLower} to the last term gives an upper bound. 

We observed numerically that the inequality $\Pr(\cap_{i\leq m}\zeta_i) \leq \Pr(\zeta_m)$ gives surprisingly good bounds if the tail of a chi-square random variable is evaluated directly. Intuitively it says that given that the $m$th transmission is in error, most of the previous error events also occur with high probability. Although equation \eqref{eqn:UpperEg} could give a better upper bound in some cases, the difference is negligible. 

Fig. \ref{fig:m5compare} shows the latency versus throughput curve for exact numerical integration, the upper bound (by lower bound on the error probability) of \eqref{eqn:LowerEg} and the lower bound (by upper bound on the error probability) of \eqref{eqn:UpperEg}  with a maximum of $5$ transmissions and optimized step sizes based on \cite{Williamson_ISIT_2012}. The channel SNR is $2$ dB and the capacity is $0.6851$. The number of information bits for each point from left to right are $16, 32, 64, 128, 256$ respectively. The bounds are sharp as promised and the upper bound on the throughput curves up when the step sizes are too small such that the lower bounds on the error probability give trivial results. Using the inequality $\Pr(\cap_{i\leq m}\zeta_i) \leq \Pr(\zeta_m)$, Fig.~\ref{fig:m5compare} also shows the lower bound on the throughput with the step size of $1$ bit, which gives the best performance among all lower bounds. 

Fig. \ref{fig:StepsCompare} shows the upper and lower bounds using \eqref{eqn:LowerEg} and \eqref{eqn:UpperEg} with different step sizes (fixed for each number of information bits) at SNR $3$ dB. The upper and lower bounds are already very tight when the step size is $10$. The upper bound on throughput is above capacity and therefore not useful when the step sizes are too small. In a practical setting where the step sizes are optimized, however, the lower bounds may still provide useful insight. Also shown in the figure is the lower bound on the throughput with the finest increment (using the inequality $\Pr(\cap_{i\leq m}\zeta_i) \leq \Pr(\zeta_m)$). The throughput of $1$-bit increment follows similar trend as in Fig.~\ref{fig:m5compare}.

\begin{figure}[t!]
\centering
\includegraphics[width=0.5\textwidth]{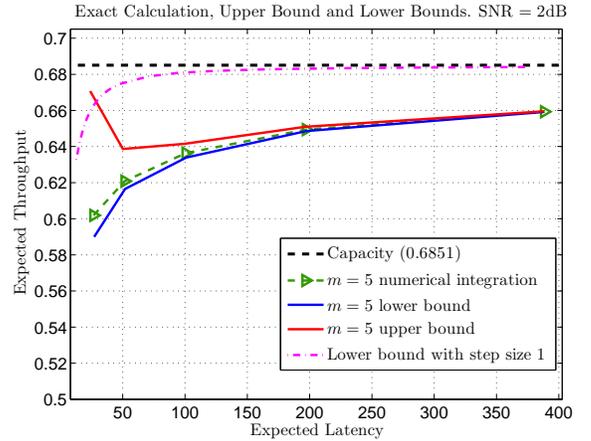}
\caption{The latency versus throughput curves of exact numerical integration, the upper and the lower bounds in the case of five transmissions on AWGN channel. Also shown in the figure is the best expected throughput by sending one bit increment at a time until the rate is well below capacity ($I_1 = \log_2M, I_i = 1, m = 3\log_2 M $, i.e., the lowest rate is $1/3$). The SNR is $2$ dB and the corresponding capacity is $0.6851$. }
\vspace*{-8pt}
\label{fig:m5compare}
\end{figure}

\begin{figure}[t!]
\centering
\includegraphics[width=0.5\textwidth]{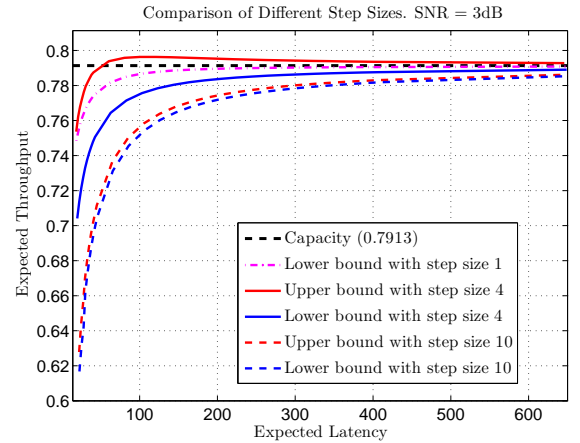}
\caption{The latency versus throughput curves of the upper and the lower bounds with difference step sizes at SNR $3$ dB. When the step sizes are small the Chernoff lower bounds on the joint error probability give trivial results and hence the throughputs go beyond the capacity. As the number of messages $M$ increases the upper and lower bound converges to the capacity. Also shown in the figure is the best expected throughput by using a step size of $1$ bit.}
\vspace*{-8pt}
\label{fig:StepsCompare}
\end{figure}

\section{Conclusions}
\label{sec:Conclusion}
This paper explores techniques to bound the performance suggested by the rate-compatible sphere-packing analysis. Using the Inglot Chernoff bounds, we derived two versions of upper and lower bounds of the relevant joint error event for the two-transmission case. The Inglot bounds are useful when the rate is slightly above capacity and the Chernoff bounds give cleaner expressions. We also presented general bounds on the $m$-transmissions case using Chernoff bounds. Numerical examples show that the bounds are tight when the step sizes are large enough. The achievable throughput with a step size of $1$ closely approaches capacity with very low latencies. This well-known yet still exciting result brings the performance of the classic coding scheme proposed by Shalkwijk and Kailath to a more practical ARQ-like coding scheme, which only requires one bit of feedback at each transmission.

We also show the finiteness of the decoding time and the expected latency using classic sphere-packing density result by Minkowski when $\eta > 1$, which also implies finiteness with perfect packing assumption.


\section*{Appendix}
This section provides the proofs of the Lemmas and Theorems in the previous sections. 
\begin{proof}[Proof of Lemma \ref{lem:ChernoffTwoTrans}]
Apply the Chernoff upper bound $\Pr(X>r) = \Pr(e^{uX}>e^{ur}) \leq \E[ e^{uX}]e^{-ur}$
to equation \eqref{eqn:Integrate1}. Let $z_i^j \equiv(z_i, z_{i+1}, \dots, z_{j}) , j > i$ and $A_{i}^j(r_k^2) = \left\{z_{i+1}^j:\|z_{i+1}^j\|^2 > r_k^2\right\}$, then
\begin{align}
&\eqref{eqn:Integrate1}\leq  \int_{r_1^2}^{\infty} \frac{\E e^{u\chi_{I_2}^2}f_{\chi_{I_1}^2}(t_1)}{e^{u(r_2^2 - t_1)}}  dt_1
\\ 
&= \int_{r_1^2}^{\infty} (1-2u)^{-I_2/2}e^{-u(r_2^2 - t_1)} f_{\chi_{I_1}^2}(t_1) dt_1
\\ \label{eqn:ChangeVar}
&=(1-2u)^{-I_2/2}e^{-u r_2^2}
\int_{A_0^{I_1}(r_1^2)}\frac{e^{-\frac{(1-2u)}{2}\sum\limits_{i = 1}^{I_1}z_i^2}}{(2\pi)^{I_1/2}} dz_1^{I_1} 
\\ \label{eqn:ChangeVar2}
&=\frac{\int_{A_0^{I_1}\left((1-2u)r_1^2\right)}\frac{e^{-\sum\limits_{i = 1}^{I_1}z_i'^2/2}}{(2\pi)^{I_1/2}} dz_1'^{I_1}}
{(1-2u)^{I_2/2}(1-2u)^{I_1/2}e^{u r_2^2}}
\\ \label{eqn:UpperBound}
&=\frac{e^{-u r_2^2}\Pr\left(\chi_{I_1}^2>(1-2u)r_1^2\right)}{(1-2u)^{N_2/2}}
\end{align}
where \eqref{eqn:ChangeVar2} follows from a change of variable $z_i' = (1+2u)^{1/2}z_i$. Taking the infimum over $u < 1/2$ gives the result. 

We sketch the proof for the lower bound due to space limitation. Observe that $\Pr(\zeta_1 \cap \zeta_2) = \Pr(\zeta_1) - \Pr(\zeta_1 \cap \zeta_2^c) = \Pr(\zeta_2) - \Pr(\zeta_1 \cap \zeta_2^c)$. Let $w_1 = \Pr(\zeta_1 \cap \zeta_2^c)$, $w_2 = \Pr(\zeta_1 \cap \zeta_2^c)$ and finding the upper bounds of them yield the lower bound. The upper bound on $w_2$ follows from the above derivation by changing the integration interval from $(r_1, \infty)$ to $[0 , r_1]$. For the upper bound on $w_1$, apply the Chernoff bound with the form $\Pr(X\leq r) = \Pr(e^{-vX}>e^{-vr}) \leq \E[ e^{-vX}]e^{vr}$. Taking the infimum over $v\geq 0$ gives the result.

\end{proof}

\begin{proof}[Proof of Theorem \ref{thm:InglotTwoTrans}]
Applying the lower bound of Theorem \ref{thm:Inglot} to equation \eqref{eqn:Integrate2} gives the lower bound. 


For the upper bound, note that the denominator of the upper bound in Theorem \ref{thm:Inglot} has a term $r_2^2 - I_2 + 2 -t$. Hence the bound can only be integrate over $r_1^2$ to $r_2^2 - I_2 + 2$. As the integration approaches $r_2^2 - I_2 + 2$ it's obvious that the bound become very loose. We may, however, split the integral into two parts and bound them separately. 
\begin{align}
	&\int_{r_1^2}^{r_2^2} \Pr(\chi_{I_2}^2 > r_2^2 - t) f_{\chi_{N_1}^2}(t) dt
	\\
	= & \int_{r_1^2}^{(1-\delta)r_2^2} \Pr(\chi_{I_2}^2 > r_2^2 - t) f_{\chi_{N_1}^2}(t)dt 
	\\
	  &+ \int_{(1-\delta)r_2^2}^{r_2^2} \Pr(\chi_{I_2}^2 > r_2^2 - t) f_{\chi_{N_1}^2}(t)dt 
	\\
	\leq & K(r_2, I_2, N_1) \int_{r_1^2}^{(1-\delta)r_2^2} g_{r_2, I_2, N_1}(t) dt 
	  + \int_{(1-\delta)r_2^2}^{r_2^2} f_{\chi_{N_1}^2}(t) dt .
\end{align}
Taking the infimum over $\delta \in (\underline{\delta}, \overline{\delta} )$ yields the upper bound:
\begin{align}
\inf_{\delta \in (\underline{\delta}, \overline{\delta} )} K \int_{r_1^2}^{(1-\delta)r_2^2} g(t) dt + \int_{(1-\delta)r_2^2}^{r_2^2} f_{\chi_{N_1}^2}(t) dt 
\end{align}
where $K, g(t)$ are functions of $r_2, I_2, N_1$:
\begin{align*}
&\overline{\delta} = \frac{r_2^2-r_1^2}{r_2^2}, \underline{\delta} = \frac{I_2-2}{r_2^2},
g_{r_2, I_2, N_1}(t) = \frac{t^{N_1/2-1}(r_2^2-t)^{I_2/2}}{r_2^2-I_2+2-t},
\\
&K(r_2, I_2, N_1) = \frac{e^{-\frac{1}{2}(r_2^2-I_2)}}{2^{N_1/2}\sqrt{\pi}I_2^{(I_2-1)/2}\Gamma(N_1/2)}.
\end{align*}
\end{proof}

\begin{proof}[Proof of Theorem \ref{thm:GeneralUpper} (sketch)]
Apply several rounds of similar steps (Chernoff upper bound and change of variable) as in the proof of Lemma \ref{lem:ChernoffTwoTrans}. 
\end{proof}



\begin{proof}[Proof of Theorem \ref{thm:DecodingTime} (sketch)]
	Given the assumption, the decoding radii $r_i$ have the following inequality:
	\begin{align}
	r_i^2 \geq \frac{c N_i(1+\eta)}{2M^{2/N_i}}, c>1.
	\end{align}	
	Note that for $\eta > 1$, $N_i/r_i^2 = c_i > 1$ for $i$ large enough. Applying the Chernoff upper bound with the optimal parameter $u_i^*$ gives a positive error exponent and the result follows from the Borel-Cantelli's lemma since $e^{-c'n}, ne^{-c'n}$ are both summable in $n$ for some $c'>0$. 
\end{proof}


\normalsize
\bibliographystyle{IEEEtran}
\bibliography{IEEEabrv,ITW_2012}

\begin{thebibliography}{10}
\providecommand{\url}[1]{#1}
\csname url@samestyle\endcsname
\providecommand{\newblock}{\relax}
\providecommand{\bibinfo}[2]{#2}
\providecommand{\BIBentrySTDinterwordspacing}{\spaceskip=0pt\relax}
\providecommand{\BIBentryALTinterwordstretchfactor}{4}
\providecommand{\BIBentryALTinterwordspacing}{\spaceskip=\fontdimen2\font plus
\BIBentryALTinterwordstretchfactor\fontdimen3\font minus
  \fontdimen4\font\relax}
\providecommand{\BIBforeignlanguage}[2]{{%
\expandafter\ifx\csname l@#1\endcsname\relax
\typeout{** WARNING: IEEEtran.bst: No hyphenation pattern has been}%
\typeout{** loaded for the language `#1'. Using the pattern for}%
\typeout{** the default language instead.}%
\else
\language=\csname l@#1\endcsname
\fi
#2}}
\providecommand{\BIBdecl}{\relax}
\BIBdecl

\bibitem{Chang_1956}
S.~Chang, ``{Theory of information feedback systems},'' \emph{{IEEE} Trans.
  Inf. Theory}, vol. PGIT-2, pp. 29--40, Sep. 1956.

\bibitem{Kramer_1969}
A.~Kramer, ``{Improving communication reliability by use of an intermittent
  feedback channel},'' \emph{{IEEE} Trans. Inf. Theory}, vol. IT-15, no.1, pp.
  52--60, Jan. 1969.

\bibitem{Zig_1970}
K.~S. Zigangirov, ``{Upper bounds for the error probability for channels with
  feedback},'' \emph{{Probl. Pered. Inform.}}, vol. 6, no.1, pp. 87--92.

\bibitem{Schalkwijk_1966_1}
J.~Schalkwijk and T.~Kailath, ``{A coding scheme for additive noise channel
  with feedback--I: No bandwidth constraint},'' \emph{{IEEE} Trans. Inf.
  Theory}, vol. IT-12, no.2, pp. 172--182, Apr. 1966.

\bibitem{Schalkwijk_1966_2}
J.~Schalkwijk, ``{A coding scheme for additive noise channel with feedback--II:
  Band-limited signals},'' \emph{{IEEE} Trans. Inf. Theory}, vol. IT-12, no.2,
  pp. 183--189, Apr. 1966.

\bibitem{Gallager_2010}
R.~G. Gallager, ``{Variations on a theme by Schalkwijk and Kailath},''
  \emph{{IEEE} Trans. Inf. Theory}, vol. 56, no.1, pp. 6--17, Jan. 2010.

\bibitem{PolyIT11}
Y.~Polyanskiy, H.~V. Poor, and S.~Verd\'{u}, ``Feedback in the non-asymptotic
  regime,'' \emph{{IEEE} Trans. Inf. Theory}, vol. 57(8), pp. 4903--4925, Aug.
  2011.

\bibitem{Chen_2010_ITA}
T.-Y. Chen, B.-Z. Shen, and N.~Seshadri, ``{Is feedback a performance equalizer
  of classic and modern codes?}'' in \emph{{ITA Workshop}}, San Diego, CA, USA,
  Feb. 2010.

\bibitem{Chen_2011_ICC}
T.-Y. Chen, N.~Seshadri, and R.~D. Wesel, ``A sphere-packing analysis of
  incremental redundancy scheme with feedback,'' in \emph{{Proc. IEEE Intl.
  Conf. Comm.}}, Kyoto, Japan, 2011.

\bibitem{Williamson_ISIT_2012}
A.~R. Williamson, T.-Y. Chen, and R.~D. Wesel, ``A rate-compatible
  sphere-packing analysis of feedback coding with limited retransmissions,'' in
  \emph{Proc. 2010 IEEE Int. Symp. Inf. Theory (ISIT)}, Jul 2012.

\bibitem{Inglot2006}
T.~Inglot and T.~Ledwina, ``Asymptotic optimality of a new adaptive test in
  regression model,'' \emph{Ann. Inst. H. Poincare}, vol.~42, pp. 579--590,
  2006.

\bibitem{ChenDraft2012}
T.-Y. Chen, A.~R. Williamson, and R.~D. Wesel, ``Rate-compatible sphere-packing
  analysis,'' \emph{draft}, 2012.

\bibitem{Sloane2002}
N.~Sloane, ``The sphere packing problem,'' \emph{arXiv:math/0207256v1
  [math.CO]}, 2002.

\end{thebibliography}

\end{document}